\documentclass[11pt]{article}

\usepackage[margin=1in]{geometry}
\usepackage[dvipsnames,usenames]{xcolor}
\usepackage{amsfonts,amsmath,amssymb,amsthm,mathtools}
\usepackage{algorithmic,algorithm}
\usepackage{bm}
\usepackage{xspace}
\usepackage[letterpaper=true,colorlinks=true,pdfpagemode=none,urlcolor=blue,linkcolor=RoyalBlue,citecolor=OliveGreen,pdfstartview=FitH]{hyperref}
\usepackage{verbatim}
\usepackage{prettyref}

\let\pref=\prettyref
\newcommand{\savehyperref}[2]{\texorpdfstring{\hyperref[#1]{#2}}{#2}}

\newrefformat{eq}{\savehyperref{#1}{\textup{(\ref*{#1})}}}
\newrefformat{lem}{\savehyperref{#1}{Lemma~\ref*{#1}}}
\newrefformat{def}{\savehyperref{#1}{Definition~\ref*{#1}}}
\newrefformat{thm}{\savehyperref{#1}{Theorem~\ref*{#1}}}
\newrefformat{cor}{\savehyperref{#1}{Corollary~\ref*{#1}}}
\newrefformat{cha}{\savehyperref{#1}{Chapter~\ref*{#1}}}
\newrefformat{sec}{\savehyperref{#1}{Section~\ref*{#1}}}
\newrefformat{app}{\savehyperref{#1}{Appendix~\ref*{#1}}}
\newrefformat{tab}{\savehyperref{#1}{Table~\ref*{#1}}}
\newrefformat{fig}{\savehyperref{#1}{Figure~\ref*{#1}}}
\newrefformat{hyp}{\savehyperref{#1}{Hypothesis~\ref*{#1}}}
\newrefformat{alg}{\savehyperref{#1}{Algorithm~\ref*{#1}}}
\newrefformat{item}{\savehyperref{#1}{Item~\ref*{#1}}}
\newrefformat{step}{\savehyperref{#1}{step~\ref*{#1}}}
\newrefformat{conj}{\savehyperref{#1}{Conjecture~\ref*{#1}}}
\newrefformat{fact}{\savehyperref{#1}{Fact~\ref*{#1}}}
\newrefformat{prop}{\savehyperref{#1}{Proposition~\ref*{#1}}}
\newrefformat{claim}{\savehyperref{#1}{Claim~\ref*{#1}}}

\newtheorem{theorem}{Theorem}%[section]
\newtheorem{lemma}[theorem]{Lemma}%[theorem]
\newtheorem{corollary}[theorem]{Corollary}
\theoremstyle{definition}
\newtheorem{definition}{Definition}
\newtheorem{remark}{Remark}

\renewcommand{\Pr}{\mathop{\bf Pr\/}}

\newcommand{\E}{\mathop{\bf E\/}}
\newcommand{\Var}{\mathop{\bf Var\/}}

\newcommand{\R}{\mathbb R}
\newcommand{\N}{\mathbb N}

\newcommand{\eqdef}{\stackrel{\textrm{def}}{=}}

\newcommand{\opt}{\mathsf{opt}}

\newcommand{\eps}{\epsilon}

\newcommand{\bb}{\boldsymbol{b}}
\newcommand{\bp}{\boldsymbol{p}}

\newcommand{\bv}{\boldsymbol{v}}

\newcommand{\abs}[1]{\left\lvert #1 \right\rvert}

\newcommand{\expo}{\textsc{Exp}}
\newcommand{\Lap}{\textsc{Lap}}

\begin{document}

\begin{titlepage}
\thispagestyle{empty}
\title{The Exponential Mechanism for Social Welfare:\\ Private, Truthful, and Nearly Optimal}

\author{Zhiyi Huang\thanks{Computer and Information Science, University of Pennsylvania. Email: {\tt hzhiyi@cis.upenn.edu}. Supported in part by ONR MURI Grant N000140710907.}
\and 
Sampath Kannan\thanks{Computer and Information Science, University of Pennsylvania. Email: {\tt kannan@cis.upenn.edu}. Supported in part by an EAGER grant, NSF CCF 1137084.}}

\maketitle

\begin{abstract}
  \thispagestyle{empty}
  In this paper we show that for {\em any} mechanism design problem with the objective of maximizing social welfare, the exponential mechanism can be implemented as a {\em truthful} mechanism while still preserving differential privacy. Our instantiation of the exponential mechanism can be interpreted as a generalization of the VCG mechanism in the sense that the VCG mechanism is the extreme case when the privacy parameter goes to infinity. To our knowledge, this is the first general tool for designing mechanisms that are both truthful and differentially private.
\end{abstract}
\end{titlepage}

\section{Introduction}

In mechanism design a central entity seeks to allocate resources among a set of selfish agents in order to optimize a specific objective function such as revenue or social welfare. Each agent has a private valuation for the resources being allocated, which is commonly referred to as her {\em type}. A major challenge in designing mechanisms for problems of resource allocation among selfish agents is getting them to reveal their true types. While in principle mechanisms can be designed to optimize some objective function even when agents are not truthful, the analysis of such mechanisms is complicated and the vast majority of mechanisms are designed to incentivize agents to be truthful.

One reason that an agent might not want to be truthful is that lying
gives her a better payoff. Research in algorithmic mechanism design has
mostly focused on this possibility and has successfully designed
computationally-efficient {\em
incentive-compatible} mechanisms for many problems  , i.e., mechanisms where each agent achieves optimal payoff by
bidding truthfully (see \cite{nisan2007algorithmic} for a survey of
results).  However, a second reason that an agent might not bid
truthfully is that the {\em privacy} of her type might itself be of
value to her. Bidding truthfully could well result in an outcome that reveals the private type of an
agent. 

Consider for example, a matching market in which $n$ oil companies are
bidding for $n$ oil fields.  A company may have done extensive
research in figuring out its valuations for each field. It may regard
this information as giving it competitive advantage and seek to protect
its privacy. If it participates in a traditional
incentive-compatible mechanism, say, the VCG mechanism, it has two
choices -- 1) bid truthfully, get the optimum payoff but potentially
reveal  private information or 2)
introduce random noise into its bid to (almost) preserve privacy, but
settle for a suboptimal payoff. In this and more generally in
multi-agent settings where each agent's type is multidimensional, we
aim to answer the following question:

\begin{quote}
    \em Can we design mechanisms that simultaneously achieve near
    optimal social welfare, are incentive compatible, and protect the
    privacy of each agent?
\end{quote}

The notion of privacy we will consider is {\em differential
privacy}, which is a paradigm for private data analysis developed
in the past decade, aiming to reveal information about the population as
a whole, while protecting the privacy
of each individual (E.g., see surveys \cite{dwork2008differential,dwarksurvey}
and the reference therein).
%Roughly speaking, a differentially private mechanism is one that behaves
%almost identically on any two data sets that are almost
%identical. Here, by behaving almost identically we mean that the
%probability of any event happening changes by at most a small
%multiplicative factor. As an important tool in the literature, the
%exponential mechanism of McSherry and Talwar
%\cite{mcsherry2007mechanism} is a general mechanism that produces
%differentially private output for a large family of problems. 
%For each problem, a quality
%value is associated with each possible answer. The exponential mechanism
%then outputs an answer with probability proportional to the exponent of
%its quality scaled by the desired differential privacy and the
%sensitivity of the answer.

\subsection*{Our Results and Techniques}

Our main contribution is a novel instantiation of the exponential mechanism for {\em any} mechanism design problem with payments, that aims to maximize social welfare. We show that our version of the exponential mechanism is incentive compatible and individually rational\footnote{Here, we consider individual rationality in expectation. Achieving individual rationality in the {\em ex-post} sense is impossible for any non-trivial private mechanism since the probability of a non-zero price would have to jump by an infinitely large factor as an agent changes from zero valuation to non-zero valuation.},
% On the one hand, if we want to obtain
%ex-post individual rationality, then we can never charge positive
%prices to any zero-value agents. On the other hand, we have to charge
%positive prices to non-zero-value agents in any non-trivial truthful
%mechanism. So the probability of charging positive prices changes by a
%multiplicative factor of infinity as an agent's valuation changes from
%all-zeroes to non-zeroes. Hence, such ex-post individually rational
%mechanisms are not private at all.}, 
%and has no positive transfer\footnote{Similar to the above.}
while preserving differential privacy.
%{\bf Did we decide not to mention positive transfer at all? Sampath}
In fact, we show that the exponential mechanism can be interpreted as a natural generalization of the VCG mechanism in the sense that the VCG mechanism is the special case when the privacy parameter goes to infinity. Alternatively, our mechanism can be viewed as an affine maximum-in-distributed-range mechanism with Shannon entropy providing the offsets. We will formally define affine maximum-in-distributed-range mechanisms in \pref{sec:prelim} and more details on this observation are deferred to \pref{sec:gibbs}. Readers are referred to \cite{dobzinski2009power, dughmi2010black, dughmi2011convex,
dughmi2011truthful} for recent applications of maximum-in-distributed-range mechanisms in algorithmic mechanism design.

%To our knowledge, this is
%the first general tool for designing truthful and differentially private
%mechanism.

%{\bf Zhiyi: In the above paragraph I add a few lines about the
%connection to maximum-in-distributed range mechanism.}

%We provide two proofs of the incentive compatibility of the exponential mechanism. The first proof (\pref{sec:expo}) is 
%We also provide another very different proof in \pref{sec:gibbs} 
Our proof is by connecting the exponential mechanism to the Gibbs measure and free energy in statistical mechanics. We exploit this connection to provide a simple proof of the incentive compatibility of the mechanism. We believe this intriguing connection is of independent interest and may lead to new ways of understanding the exponential mechanism and differential privacy.
%We also present another very different proof using the classical characterization of when an allocation mechanism can be associated with prices to make it incentive-compatible. Rochet \cite{rochet1987necessary} showed that this is possible exactly in the case that the mechanism is cyclic monotone. In \pref{app:rochet}, we prove that the exponential mechanism is cyclic monotone and derive the payments according to Rochet's characterization. 

While we do not have an efficient way of computing the allocation and prices of the exponential mechanism in general (this is also not known for VCG), we do show that in special cases such as multi-item auctions and procurement auctions for spanning tree, we can efficiently implement the exponential mechanism either exactly or approximately. Further, we show that the trade-off between privacy and social welfare in the exponential mechanism is asymptotically optimal in these two cases, even if we compare to mechanisms that need not be truthful. We also include another application of the exponential mechanism for the combinatorial public project problem where the social welfare  is close to optimal for an arbitrarily small constant $\epsilon$.

Interestingly, our implementation of the exponential mechanism for multi-item auctions has further implications in the recent work on blackbox reductions in Bayesian mechanism design \cite{hartline2011bayesian, bei2011bayesian}. Combining our exponential mechanism for the matching market with the blackbox reduction procedure in \cite{hartline2011bayesian, bei2011bayesian}, we can get a blackbox reduction that converts any algorithm into BIC, differentially private mechanisms. %without hurting the social welfare too much. 
We will leave further discussions to the relevant section.

\subsection*{Related Work} 

McSherry and Talwar \cite{mcsherry2007mechanism} first proposed using differentially private mechanisms to design auctions by pointing out that differential privacy implies approximate incentive compatibility as well as resilience to collusion. In particular, they study the problem of revenue maximization in digital auctions and attribute auctions. They propose the exponential mechanism as a solution for these problems. McSherry and Talwar also suggest using the exponential mechanism to solve mechanism design problems with different objectives, such as social welfare.\footnote{The main difference between our instantiation of the exponential mechanism and that by McSherry and Talwar is that we use properly chosen payments to incentivize agents to report truthfully.} 
Their instantiation of the exponential mechanism is differentially private, but only approximately truthful. Nissim et al.~\cite{nissim2012approximately} show how to convert differentially private mechanisms into exactly truthful mechanism in some settings. However, the mechanism loses its privacy property after such conversion.  Xiao \cite{xiao2011privacy} seeks to design mechanisms that are both differentially private and perfectly truthful and proposes a method to convert any truthful mechanism into a differentially private and truthful one when the type space is small. Unfortunately, it does not seem possible to extend the results in \cite{nissim2012approximately, xiao2011privacy} to more general mechanism design problems, while our result applies to {\em any} mechanism design problem (with payments).

Xiao \cite{xiao2011privacy} also proposed to explicitly model the agents' concern for privacy in the utilities by assuming agent $i$ has a disutility that depends on the amount of information $\epsilon_i$ leaked by the mechanism. Chen et al. [7] and Nissim et al. [25] explored this direction and introduced truthful mechanisms for some specific problems. Exact evaluation of an agent's dis-utility usually requires knowledge of the types of all agents and hence this kind of mechanism can only be private if agents do not need to exactly compute their own dis-utility. The above works circumvent this issue by designing strictly truthful and sufficiently private mechanisms such that any agent's gain in privacy by lying is outweighed by the loss in the usual notion of utility, regardless of the exact value of dis-utility for privacy.

Finally, Ghosh and Roth \cite{ghosh2011selling} study the problem of
selling privacy in auctions, which can be viewed as an orthogonal
approach to combining mechanism design and differential privacy.

\section{Preliminaries}
\label{sec:prelim}

%\paragraph{Model}

A mechanism design problem is defined by a set of $n$ agents and a range $R$ of feasible outcomes. Throughout this paper we will assume the range $R$ to be discrete, but all our results can be easily extended to continuous ranges with appropriate integrability. Each agent $i$ has a private valuation $v_i : R \mapsto [0, 1]$. A central entity chooses one of the outcomes based on the agents' (reported) valuations. We will let $\bm{0}$ denote the all-zero valuation and let $v_{-i}$ denote the valuations of every agent except $i$.

For the sake of presentation, we will assume that
the agents' valuations can be any functions mapping the range of
feasible outcomes to the interval $[0, 1]$. It is worth noting that since
our mechanisms are incentive compatible in this setting, they are also 
automatically incentive compatible for more restricted valuations
(e.g., submodular valuations for a combinatorial public project problem).

A {\em mechanism} $M$ consists of an allocation rule $x(\cdot)$ and a
payment rule $p(\cdot)$. The mechanism first lets the agents submit
their valuations. However, an agent may strategically submit a false
valuation if that is beneficial to her. We will let $b_1, \dots, b_n: R
\mapsto [0, 1]$ denote the {\em reported valuations} (bids)  from the agents and
let $\bb$ denote the vector of these valuations. 
After the agents submit their bids, the allocation rule $x(\cdot)$
chooses a feasible outcome $r = x(\bb) \in R$ and the payment rule
$p(\cdot)$ chooses a vector of payments $p(\bb) \in \R^n$. 
We will let $p_i(\bb)$ denote the payment for agent $i$. 
Note that both $x(\cdot)$ and $p(\cdot)$ may be randomized. 
We will consider the standard setting of quasi-linear utility:
given the
allocation rule, the payment rule, and the reported valuations $\bb$,
for each $i \in [n]$, the {\em utility} of agent $i$ is 
$$u_i(v_i, x(\bb), p_i(\bb)) = v_i(x(\bb)) - p_i(\bb) \enspace.$$
We will assume the agents are risk-neutral and aim to maximize their expected utilities.

The goal is to design polynomial time mechanisms
$M$ that satisfy various objectives. In this paper, we will focus on
the problem of maximizing the expected {\em social
welfare}, which is defined to be the sum of the agents' valuations:
$\E[\sum_{i=1}^n v_i(x(\bb))]$.

Besides the expected social welfare, we take into consideration
the strategic play of utility-maximizing agents and their concern
about the mechanism leaking non-trivial information about their private
data. Thus, we will restrict our attention to mechanisms that
satisfy several
game-theoretic requirements and have a privacy guarantee that we
will define in the rest of this section.

\subsection{Game-Theoretical Solution Concepts} 

A mechanism is {\em incentive compatible} (IC) if truth-telling is a
dominant strategy, i.e.,  by reporting
the true values an agent always maximizes her expected utility
regardless of what other agents do - $v_i
\in \arg\max_{b_i} \E[v_i(x(b_i, b_{-i})) - p_i(b_i, b_{-i})]$.
We will also consider an approximate notion of truthfulness. 
A mechanism is {$\gamma$-incentive compatible} ($\gamma$-IC) if no agent
can get more than $\gamma$ extra utility by lying.
Further, a mechanism is {\em individually rational} (IR) if the {\em expected}
utility of each agent is always non-negative, assuming this agent
reports truthfully: $\E[v_i(x(v_i, b_{-i})) - p_i(v_i, b_{-i})] \ge 0$.
%Finally, a mechanism has {\em no positive transfer} if the {\em expected} payments are always non-negative: $\forall b_1, \dots, b_n, \forall i \in [n], \E [p(\bb)_i] \ge 0$. 
We seek to design mechanisms that are incentive compatible and
individually rational.
%, and without positive transfer.

\paragraph{Affine Maximum-In-Distributed-Range}

An allocation rule $x(\cdot)$ is an {\em affine maximum-in-distributed-range allocation} if there is a set $S$ of distributions over feasible outcomes, parameters $a_1, \dots, a_n \in \R^+$, and an offset function $c : S \mapsto \R$, such that the $x(v_1, \dots, v_n)$ always chooses the distribution $\nu \in S$ that maximizes
$$\E_{r \sim \nu} \left[ \sum^n_{i = 1} a_i v_i(r) \right] + c(\nu) \enspace.$$
In this paper, we are particularly interested in the case when $a_i = 1$, $\forall i \in [n]$, and $c$ is the Shannon entropy of the distribution scaled by an appropriate parameter.

The affine maximum-in-distributed-range mechanisms can be interpreted as
slight generalizations of the well-studied maximum-in-distributed-range
mechanisms.
If $a_i = 1$ for every $i \in [n]$ and $c(\cdot) = 0$, then such allocation rules are referred to as maximum-in-distributed-range (MIDR) allocations. There are well-known techniques for charging proper prices to make MIDR allocations and their affine generalizations incentive compatible. The resulting mechanisms are called MIDR mechanisms. MIDR mechanisms are important tools for designing computationally efficient mechanisms that are incentive compatible and approximate social welfare well (e.g., see \cite{dobzinski2009power, dughmi2010black, dughmi2011truthful, dughmi2011convex}).

\subsection{Differential Privacy} 

Differential privacy is a notion of privacy
that has been studied the most in the theoretical computer science community over the past decade. It
requires the distribution of outcomes to be nearly identical when the
agent profiles are nearly identical. Formally,

\begin{definition}
    A mechanism is {\em $\epsilon$-differentially private} if for any
    two valuation profiles $\bv = (v_1, \dots, v_n)$ and $\bv' = (v_1',
    \dots, v_n')$ such that only one agent has different valuations in
    the two profiles, and for any set of outcomes $S \subseteq R$, we
    have 
    $$\Pr[x(\bv) \in S] \le \exp(\epsilon) \cdot \Pr[x(\bv')
    \in S] \enspace.$$
\end{definition}

This definition of privacy has many appealing theoretical properties.
Readers are referred to \cite{dwork2008differential, dwarksurvey}
for excellent surveys on the subject.

%We
%want to stress that this assumption is w.l.o.g.~for, by adding
%arbitrary noise with zero mean we can obtain a payment scheme that is
%almost perfectly private without affecting our objective or any of the
%game-theoretic requirements.

We will also consider a standard variant that defines
a more relaxed notion of privacy.

\begin{definition}
    A mechanism is {\em $(\epsilon, \delta)$-differentially private} if
    for any two valuation profile $\bv = (v_1, \dots, v_n)$ and $\bv' =
    (v_1', \dots, v_n')$ such that only one agent has different
    valuations in the two profiles, and for any set of outcomes $S
    \subseteq R$, $$\Pr[x(\bv) \in S] \le \exp(\epsilon) \cdot
    \Pr[x(\bv') \in S] + \delta \enspace.$$
\end{definition}

Typically, we will consider very small values of $\delta$, say, $\delta =
\exp(-n)$. 
%This relaxed notion of differential privacy
%allows the  probability of very low-probability events to be sensitive to
%the change of a single agent's valuation.

\paragraph{Differentially Private Payment}

In the above definitions, we only consider the privacy of the allocation rule. We note that in practice, the payments need to be differentially private as well. We can handle privacy issues in the payments by the standard technique of adding Laplace noise. In particular, if the payments are implemented via secure channels (e.g., the same channels that the agents use to submit their bids) such that the each agent's payment is accessible only by the agent herself and the central entity, then adding independent Laplace noise with standard deviation $O(\epsilon^{-1})$ is sufficient to guarantee $\epsilon$-differentially private payments. Since the techniques used to handle payments are quite  standard, we will defer the extended discussion of this subject to the appendix.%\pref{app:payment}.

\subsection{The Exponential Mechanism}

One powerful tool in the differential privacy literature is the exponential mechanism of McSherry and Talwar \cite{mcsherry2007mechanism}. The exponential mechanism is a general technique for constructing differentially private algorithms over an arbitrary range $R$ of outcomes and any objective function $Q(D, r)$ (often referred to as the quality function in the differential privacy literature) that maps a pair consisting of a data set $D$ and a feasible outcome $r \in R$ to a real-valued score. In our setting, $D$ is a (reported) valuation profile and the quality function $Q(\bv, r) = \sum_{i=1}^n v_i(r)$ is the social welfare. 

Given a range $R$, a data set $D$, a quality function $Q$, and a privacy parameter $\epsilon$, the {\em exponential mechanism} $\expo(R, D, Q, \epsilon)$ chooses an outcome $r$ from the range $R$ with probability 
$$\Pr\left[\expo(R, D, Q, \epsilon) = r\right] \propto \exp\left(\frac{\epsilon}{2\Delta} Q(D, r)\right) \enspace,$$
where $\Delta$ is the Lipschitz constant of the quality function $Q$, that is, for any two adjacent data set $D_1$ and $D_2$, and for any outcome $r$, the score $Q(D_1, r)$ and $Q(D_2, r)$ differs by at most $\Delta$. In out setting, the Lipschitz constant of the social welfare function is $1$. We sometimes use $\expo(D, \epsilon)$ for short when the range $R$ and the quality function $Q$ is clear from the context. We will use the following theorem about the exponential mechanism. 

\begin{theorem}[E.g., \cite{mcsherry2007mechanism,talwardifferentially}] \label{thm:expo}
    The exponential mechanism is $\epsilon$-differentially private and
    ensures that
    $$\Pr \left[ Q(D, \expo(D, \epsilon)) < \max_{r \in R} Q(D, r) - \frac{\ln\abs{R}}{\epsilon} - \frac{t}{\epsilon} \right] \le \exp(-t) \enspace.$$
\end{theorem}

\section{The Exponential Mechanism is Incentive Compatible} 
\label{sec:expo}

In this section, we will show that if we choose the social welfare to
be the quality function, then the exponential mechanism can be
implemented in an incentive compatible and individually rational
%, and no-positive-transfer 
manner. Formally, for any range $R$ and any privacy
parameter $\epsilon > 0$, the exponential mechanism $\expo^R_\epsilon$
with its pricing scheme is presented in \pref{fig:expo}. Our main
theorem is the following:

%{\bf  Sampath: "Our instantiation" is used too often in the paper without explanation. Explain it the first time it is used, but try not to use this phrase again. Need to say that S is the Shannon entropy. Also need to decide whether we are going to talk about "no positive transfer" or not and do this consistently throughout the paper.}

\begin{figure*}[!t]
    \centering
    \fbox{\begin{minipage}{.98\textwidth}
        \begin{enumerate}
            \item Choose outcome $r \in R$ with probability 
                $\Pr[r] \propto \exp\left(\frac{\epsilon}{2} \sum_i v_i(r)
                \right)$.
            \item For $1 \le i \le n$, charge agent $i$ price
                $$p_i = - \E_{r \sim \expo^R_\epsilon(b_i,
                b_{-i})} \left[ \sum_{k \ne i} b_k(r) \right] -
                \frac{2}{\epsilon} \cdot S \left( \expo^R_\epsilon(b_i,
                b_{-i}) \right) + \frac{2}{\epsilon} \ln \left( \sum_{r
                \in R} \exp \left( \frac{\epsilon}{2} \sum_{k \ne i}
                v_k(r) \right) \right) \enspace,$$
                where $S(\cdot)$ is the Shannon entropy.
                %does not depend on agent $i$'s bid.
                %$$p_i = \E_{r \sim \expo^R_\epsilon(\bv)} [v_i(r)] - \frac{2}{\epsilon} \ln \left(\sum_{r \in R} \exp\left( \frac{\epsilon}{2} \sum_{k = 1}^n v_k(r) \right) \right) + \frac{2}{\epsilon} \ln \left( \sum_{r \in R} \exp \left( \frac{\epsilon}{2} \sum_{k \ne i} v_k(r) \right) \right) \enspace.$$
        \end{enumerate}
    \end{minipage}}
    \caption{$\expo^R_\epsilon$: the incentive-compatible exponential
    mechanism.}
    \label{fig:expo}
\end{figure*}

\begin{theorem} \label{thm:main}
    The exponential mechanism with our pricing scheme is IC and IR.
    %incentive compatible, individually rational, and has no positive transfer.
\end{theorem}

Our proof of \pref{thm:main} relies on the connection between the exponential mechanism and a well known probability measure in probability and statistical mechanics called the Gibbs measure. Once we have established this connection, the proof of \pref{thm:main} becomes very simple.

\subsection{The Exponential Mechanism and the Gibbs Measure}
\label{sec:gibbs}

The Gibbs measure, also known as the Boltzmann distribution in chemistry and physics, is formally defined as follows:

\begin{definition}[Gibbs measure]
    Suppose we have a system consisting of particles of a gas. If the particles have $k$ states $1, \dots, k$, possessing energy $E_1, \dots, E_k$ respectively, then the probability that a random particle in the system has state $i$ follows the {\em Gibbs measure}:
    $$\Pr[\text{state} = i] \propto \exp\left( - \frac{1}{k_B T} E_i \right) \enspace,$$
    where $T$ is the temperature, and $k_B$ is the Boltzmann constant.
\end{definition} 

Note that the Gibbs measure asserts that nature prefers states with lower energy level. Indeed, if $T \rightarrow 0$, then almost surely we will see a particle with lowest-energy state. On the other hand, if $T \rightarrow +\infty$, then all states are equally likely to appear. Thus the temperature $T$ is a measure of uncertainty in the system: the lower the temperature, the less uncertainty  in the system, and vice versa. 

\paragraph{Gibbs Measure vs.~Exponential Mechanism} 
It is not difficult to see the analogy between the Gibbs measure and the
exponential mechanism. Firstly, the quality $Q(r)$ of an outcome $r \in
R$ (in our instantiation, $Q(r)$ is the social welfare $\sum_i v_i(r)$)
is an analog of the energy (more precisely, the negative of the
energy) of a state $i$. In the exponential mechanism the goal is to
maximize the expected quality of the outcome, while in physics 
nature tries to minimize the expected energy. Second, the privacy
parameter $\epsilon$ is an analogue of the inverse temperature
$T^{-1}$, both measuring the level of uncertainty in the system. 
The more privacy we want in the mechanism, the
more uncertainty we need to impose in the distribution of
outcomes\footnote{We note that the privacy guarantee $\epsilon$ is not
necessarily a monotone function of the entropy of the outcome
distribution. So the statement above is only for the purpose of establishing a high-level connection
between the Gibbs measure and the exponential mechanism.}. Finally, the
Lipschitz constant $\Delta$ and Boltzmann constant $k_B$ are both
scaling factors that come from the environment. \pref{tab:gibbs}
summarize this connection between the Gibbs measure and the exponential
mechanism.

%\begin{table}[!t]
%% increase table row spacing, adjust to taste
%\renewcommand{\arraystretch}{1.3}
% if using array.sty, it might be a good idea to tweak the value of
% \extrarowheight as needed to properly center the text within the cells
%\caption{An Example of a Table}
%\label{table_example}
%\centering
%% Some packages, such as MDW tools, offer better commands for making tables
%% than the plain LaTeX2e tabular which is used here.
%\begin{tabular}{|c||c|}
%\hline
%One & Two\\
%\hline
%Three & Four\\
%\hline
%\end{tabular}
%\end{table}

\begin{table*}[!t]
    \caption{A high-level comparison between the Gibbs measure and the exponential mechanism} 
    \label{tab:gibbs}
    \centering
    \medskip
    \begin{tabular}{ccc}
        \hline \\ [-1.8ex]
        & Gibbs measure & Exponential mechanism \\ [.7ex]
        \hline \\ [-1.8ex]
        \text{Probability mass function} & 
        $\Pr[\text{state} = i] \propto \exp\left( - \frac{1}{k_B T} E_i
        \right)$ & $\Pr[\text{outcome} = r] \propto \exp\left(
        \frac{\epsilon}{2\Delta} Q(r)\right)$ \\ [.7ex]
        \text{Objective function} & $-E_i$ & $Q(r)$ \\ [.7ex]
        \text{Measure of uncertainty} &
        temperature $T$ & privacy parameter $\epsilon$ \\ [.7ex]
        \text{Environment parameter} & Boltzmann constant $k_B$ &
        Lipschitz constant $\Delta$ \\ [.7ex]
        \hline
    \end{tabular}
\end{table*}

\paragraph{Gibbs Measure Minimizes Free Energy} 

It is well-known that the Gibbs measure maximizes entropy given the expected energy. In fact, a slightly stronger claim (e.g., see \cite{le2008introduction}) states that the Gibbs measure minimizes free energy. To be precise, suppose $T$ is the temperature, $\nu$ is a distribution over the states, and $S(\nu)$ is the Shannon entropy of $\nu$. Then the {\em free energy} of the system is 
$$F(\nu, T) = \E_{i \sim \nu}[E_i] - k_B T \cdot S(\nu) \enspace.$$
The following result is well known in the statistical physics literature. 
%For completeness, we include the proof in \pref{app:gibbs}.

\begin{theorem}[E.g., see \cite{le2008introduction}]
    \label{thm:freeenergy}
    $F(\nu, T)$ is minimized when $\nu$ is the Gibbs measure.
\end{theorem}

For self-containedness, we include the proof of \pref{thm:freeenergy} as
follows.

\begin{proof}
    Note that the free energy can be written as 
    \begin{eqnarray} 
        F(\nu, T) & = & \E_{i \sim \nu}[E_i] - k_B T \cdot S(\nu) \notag \\
        & = & \sum_i \Pr_\nu[i] E_i + k_B T \sum_i \Pr_\nu[i] \ln \Pr_\nu[i] \enspace. \label{eq:2}
    \end{eqnarray}

    Further, the first term of the right hand side can be rewritten as
    \begin{eqnarray}
        \sum_i \Pr_\nu[i] E_i & = & k_B T \sum_i \Pr_\nu[i] \frac{1}{k_B T} E_i \notag \\
        & = & - k_B T \sum_i \Pr_\nu[i] \ln\left(\exp\left(-\frac{1}{k_B T} E_i \right)\right) \notag \\
        & = & - k_B T \sum_i \Pr_\nu[i] \ln \left( \frac{\exp \left( - \frac{1}{k_B T} E_i \right)}{\sum_j \exp \left( - \frac{1}{k_B T} E_j \right)} \right) - k_B T \ln \left( \sum_j \exp \left( - \frac{1}{k_B T} E_j \right) \right) \notag \\
        & = & - k_B T \sum_i \Pr_\nu[i] \ln \left( \Pr_{\textit{\rm Gibbs}}[i] \right) - k_B T \ln \left( \sum_j \exp \left( - \frac{1}{k_B T} E_j \right) \right) \enspace. \label{eq:3}
    \end{eqnarray}

    By \pref{eq:2} and \pref{eq:3}, the free energy equals
    $$F(\nu, T) =  k_B T \cdot D_{KL}(\nu \, || \, \textit{\rm Gibbs}) - k_B T \ln \left( \sum_j \exp \left( - \frac{1}{k_B T} E_j \right) \right) \enspace.$$    
    Note that the second term is independent of $\nu$. By basic properties of the KL-divergence, the above is minimized when $\nu$ is the Gibbs measure.
\end{proof}

\subsection{Proof of \pref{thm:main}}

By the connection between Gibbs measure and exponential mechanism and \pref{thm:freeenergy}, we have the following analogous lemma for our instantiation of the exponential mechanism.

\begin{lemma} \label{lem:freeenergy}
    The {\em free social welfare}, 
    $$\E_{r \sim \nu}\left[\sum_i v_i(r)\right] + \frac{2}{\epsilon} \cdot S(\nu) ~,$$ 
    is maximized when $\nu = \expo^R_\epsilon(v_1, \dots, v_n)$.
\end{lemma}

\paragraph{Incentive Compatibility} 

Let us consider a particular agent $i$, and fix the bids $b_{-i}$ of the
other agents. Suppose agent $i$ has value $v_i$ and bids $b_i$. For
notational convenience, we let $b(r) = \sum_{k=1}^n b_k(r)$ and let
$$h_i(b_{-i}) = \frac{2}{\epsilon} \ln \left( \sum_{r \in R}
\exp \left( \frac{\epsilon}{2} \sum_{k \ne i} v_k(r) \right) \right) \enspace.$$
Using the price $p_i$ charged to agent $i$ as in \pref{fig:expo}, her utility when she bids $b_i$ is
$$\E_{r \sim \expo^R_\epsilon(b_i, b_{-i})}[v_i(r) + \sum_{k \ne i} b_k(r)] + \frac{2}{\epsilon} \cdot S(\expo^R_\epsilon(b_i, b_{-i})) - h_i(b_{-i}) \enspace,$$
which equals the free social welfare plus a term that does not depend on
agent $i$'s bid. By \pref{lem:freeenergy}, the free social welfare is
maximized when we use the outcome distribution by the exponential
mechanism with respect to agent $i$'s true value. Therefore, 
truthful bidding is a utility-maximizing strategy for agent $i$.

\paragraph{Individual Rationality} We first note that for any agent
$i$, it is not difficult to verify that $p_i = 0$ when $v_i = \bm{0}$
regardless of bidding valuations of other agents. Therefore, by bidding
$\bm{0}$ agent $i$ could always guarantee non-negative expected utility.
Since we have shown that the exponential mechanism is
truthful-in-expectation, we get that the utility of agent $i$ when she
truthfully reports her valuation is always non-negative.

\begin{remark}
    We notice that \pref{lem:freeenergy} implies that the allocation
    rule of the exponential mechanism is affine
    maximum-in-distributed-range. As a result, there are standard
    techniques to charge prices so that the mechanisms is IC and IR as
    presented above. 
\end{remark}

%{\bf Sampath: Add a new remark that says that one could also prove this via Rochet's characterization; in particular the expo mech is cyclic monotone.}
%\begin{comment}
\begin{remark}
    Alternatively, one can prove \pref{thm:main} via the procedure
    developed by Rochet \cite{rochet1987necessary}: first prove the
    cyclic monotonicity of the exponential allocation rule, which is
    known to be the necessary and sufficient condition for being the
    allocation rule of a truthful mechanism; then derive the pricing
    scheme that rationalizes the exponential allocation rule via
    Rochet's characterization. We will omit further details of this
    proof in this extended abstract.
    %include this proof in \pref{app:rochet}.  
\end{remark}
%\end{comment}

\section{Generalization}
\label{sec:general}

In the original definition by McSherry and Talwar \cite{mcsherry2007mechanism}, the exponential mechanism is defined with respect to a prior distribution $\mu(\cdot)$ over the feasible range $R$. More precisely, the exponential mechanism given $\mu$, $\expo_{\mu}(R, D, Q, \epsilon)$, chooses an outcome $r$ from the range $R$ with probability 
$$\Pr \left[ \expo_{\mu}(R, D, Q, \epsilon) = r \right] \propto \mu(r) \exp \left(\frac{\epsilon}{2\Delta} Q(D, r) \right) \enspace.$$

When $\mu$ is chosen to be the uniform distribution over the feasible range, we recover the definition in \pref{sec:prelim}. Using a different $\mu$ can improve computational efficiency as well as the trade-off between privacy and the objective for some problems (e.g., \cite{blum2008learning}). In every use of the (generalized) exponential mechanism, to our knowledge, $\mu$ is taken to be the uniform distribution over a sub-range that forms a geometric covering of the feasible range. But in general, this need not be the optimal choice. 

We observe that our result can be extended to the above generalized exponential mechanism as well. More precisely, we can show that the generalized exponential mechanism is affine maximum-in-distributed-range as well.

\begin{theorem} \label{thm:general}
  For any range $R$, any quality function $Q$, any privacy parameter $\epsilon$, any prior distribution $\mu$, and any database $D$, the generalized exponential mechanism satisfies
  $$\expo_{\mu}(R, D, Q, \epsilon) = \arg\max_{\nu} \E_{r \sim \nu}[Q(D, r)] - \frac{2}{\epsilon} D_{KL}(\nu || \mu) \enspace.$$
\end{theorem}

\begin{corollary} \label{cor:general}
  For any mechanism design problem for social welfare and any prior distribution $\mu$ over the feasible range, the generalized exponential mechanism (w.r.t.~$\mu$) is IC and IR with appropriate payment rule.
\end{corollary}

The proof of \pref{thm:general} and deriving the pricing scheme in \pref{cor:general} is very similar to the corresponding parts in \pref{sec:expo} and hence omitted.

\section{Applications}

Our result in \pref{thm:main} applies to a large family of problems. 
In fact, it can be used to derive truthful and differentially private
mechanisms for any problem in mechanism design (with payments) that aims
for social welfare maximization.

In this section, we will consider three examples -- the combinatorial
public project problem (CPPP), the multi-item auction, and the
procurement auction for a spanning tree. The exponential mechanism for
the combinatorial public project problem is incentive compatible,
$\epsilon$-differentially private, and achieves nearly optimal social
welfare for any constant $\epsilon > 0$. However, we cannot implement
the exponential mechanism in polynomial time for CPPP in general because
implementing VCG for CPPP is known to be $\textit{\bf NP}$-hard and the
exponential mechanism is a generalization of VCG. For the other two
applications, we manage to implement the exponential mechanism in
polynomial time, where the implementation for multi-item auction is only
approximate so that it is only approximately truthful and approximately
differentially private, and the implementation for procurement auction
for spanning trees is exact. The social welfare for these two cases,
however, is nearly optimal only when the privacy parameter $\epsilon$ is
super-constantly large. Nonetheless, we show that the trade-offs
between privacy and social welfare of the exponential mechanism in these
two applications are asymptotically optimal.
%We will analyze the computational efficiency, including how to
%efficiently choose an outcome from the desired distribution and how to
%efficiently generate the payment. We will also consider the trade-off
%between social welfare and privacy of our instantiation of the
%exponential mechanism.

\subsection{Combinatorial Public Project Problem}

The first interesting application of our result is a truthful and
differentially private mechanism for the Combinatorial Public Project
Problem (CPPP) originally proposed by Papadimitriou et
al.~\cite{papadimitriou2008hardness}. In
CPPP, there are $n$ agents and $m$ public projects. Each agent $i$ has a
private valuation function $v_i$ that specifies agent $i$'s value
(between $0$ and $1$) for every subset of public projects. The objective
is to find a subset $S$ of public projects  to build, of size at most $k$ (a
parameter), that maximizes the social welfare, namely, $\sum_i
v_i(S)$.

This problem has received a lot of attention in the algorithmic game
theory literature because strong lower bounds can be shown for the
approximation ratio of this problem by any truthful mechanism when the
valuations are submodular (e.g., see \cite{papadimitriou2008hardness,
dughmi2011limitations}).

Further, the CPPP is of practical interest as well. The following is a
typical CPPP scenario in the real world. Suppose some central entity
(e.g., the government) wants to
build several new hospitals where there are $m$ potential locations to
choose from. Due to resource constraints, the government can only build
$k$ hospitals. Each citizen has a private value for each subset
of locations that may depend on the distance to the closest hospital
and the citizen's health status. 

Note that the agents may be concerned about their privacy if they choose to
participate in the mechanism because their valuations typically contain
sensitive information. For
example, the citizens who have high values for having a hospital close by
in the above scenario are more likely to have health problems.
Therefore, it would be interesting to design mechanisms for the CPPP that
are not only truthful but also differentially private. The
size of the range of outcomes is ${m \choose k} = O(m^k)$. So by
\pref{thm:expo} and \pref{thm:main}, we have the following.

\begin{theorem} \label{thm:cppp}
    For any $\epsilon > 0$, the exponential mechanism
    $\expo^{\textit{\rm CPPP}}_{\epsilon}$ for CPPP is IC,
    $\epsilon$-differentially private, and ensures 
    $$\Pr\left[\sum_{i=1}^n v_i \left(\expo^{\textit{\rm
    CPPP}}_\epsilon \right) < \opt - \frac{k \ln m}{\epsilon} -
    \frac{t}{\epsilon} \right] \le \exp(-t) \enspace.$$
\end{theorem}

It is known that the exponential mechanism achieves the optimal trade-off between privacy and social welfare for CPPP (e.g., \cite{talwardifferentially}).

Further, note that the optimal social welfare could be as large as $n$. Moreover, the number of projects $k \le m$ is typically much smaller than the number of agents $n$. Therefore, the exponential mechanism achieves social welfare that is close to optimal. However, it is worth noting  that we only requires $k$ and $m$ to be mildly smaller than $n$ (e.g., $O(n^{1-c})$ for any small constant $c > 0$), in which cases the size of the type space, which is exponential in $k$ and $m$, is still quite large so that the approach in \cite{xiao2011privacy} does not apply.

In some scenarios such as the one above  where the government wants to
build a few new hospitals, $k$ is sufficiently small so that it is
acceptable to have running time polynomial in the size of the range of
outcomes. In such cases, it is easy to see that the exponential
mechanism for CPPP can be implemented in time polynomial in $n$ and
${m \choose k}$. 
%Unfortunately, we cannot implement the exponential
%mechanism in time polynomial in $n$, $m$, and $k$ for the general case
%because the exponential mechanism is a generalization of the VCG
%mechanism, which is known to be {\bf NP}-hard to implement for the CPPP. 

\subsection{Multi-Item Auction}

Next we consider a multi-item auction. Here,
 the auctioneer has $n$ heterogeneous
items (one copy of each item) that she wishes to allocate to $n$
different agents\footnote{The case when the number of items is not
the same as the number of agents can be reduced to this case by adding
dummy items or dummy agents. So our setting is w.l.o.g.}. Agent $i$
has a private valuation $\bv_i = (v_{i1}, \dots, v_{ik})$, where
$v_{ij}$ is her
value for item $j$. We will assume the agents are unit-demand, that is,
each agent wants at most one item. It is easy to see that each
feasible allocation of the multi-item auction is a matching between
agents and items. We will let the $R_M$ denote the range of multi-item
auction, that is, the set $\Pi_n$ of all permutations on $[n]$.

The multi-item auction and related problems are
very well-studied in the algorithmic game theory literature
(e.g., \cite{chawla2009sequential,bhattacharya2010budget}).
They capture the motivating scenario of 
allocating oil fields and many other problems that arise from
allocating public resources.
The VCG mechanism can be implemented in polynomial time to
maximize social welfare in this problem since max-matching can be solved
in polynomial time. The new twist in our setting is 
to design mechanisms that are {\em both truthful and differentially
private} and have good social welfare guarantee.

\paragraph{Approximate Implementation of the Exponential Mechanism}
Unfortunately, exactly sampling matchings according to the distribution
specified in the exponential mechanism seems hard due to its
connection to the problem of computing the permanent of non-negative
matrices (e.g., see \cite{jerrum1989approximating}), which is
$\#P$-complete. Instead, we will
sample from the desired distribution approximately. Moreover, we show
that there is an efficient approximate implementation of the payment
scheme. As a result of the non-exact implementation, we only get
$\gamma$-IC instead of perfect IC, $(\epsilon, \delta)$-differential
privacy instead of $\epsilon$-differential privacy, and lose an
additional $n \gamma$ additive factor in social welfare. Here, $\gamma$
will be inverse polynomially small. The discussion of this approximate
implementation of the exponential mechanism is deferred to the full version.

\medskip 

Note that the size of the range of feasible outcomes of multi-item
auction is $n!$. By \pref{thm:expo}, we have the following:

\begin{theorem} \label{thm:matching}
    For any $\delta \in (0, 1)$, $\epsilon > 0$, $\gamma > 0$, there is
    a polynomial time (in $n$, $\epsilon^{-1}$, $\gamma^{-1}$, and
    $\log(\delta^{-1})$) approximate implementation of the exponential
    mechanism, $\widehat{\expo}^{R_M}_\epsilon$ that is
    $\gamma$-IC, $(\epsilon,\delta)$-differentially private, and
    ensures that 
    $$\Pr\left[\sum_{i=1}^n v_i \left( \widehat{\expo}^{R_M}_\epsilon \right) < \opt - \gamma n - \frac{\ln(n!)}{\epsilon} - \frac{t}{\epsilon} \right] \le \exp(-t) \enspace.$$
\end{theorem}

Note that here we are achieving $\gamma$-IC and $(\epsilon, \delta)$-differentially privacy while in the instantiation of the exponential mechanism by McSherry and Talwar \cite{mcsherry2007mechanism} is $\epsilon$-IC and $\epsilon$-differentially private. Our result in \pref{thm:matching} is better in most applications since typically $\eps$ is large, usually a constant or occasionally a super-constant, while $\gamma$ is small, usually requires to be $1 / \textrm{poly}$ for $\gamma$-IC to be an appealing solution concept.
%{\bf Sampath: Is delta the same as gamma or different? The discussion above the theorem suggests they are the same but the theorem treats them as different.}

The trade-off between privacy and social welfare in \pref{thm:matching} can be interpreted as the follows: if we want to achieve social welfare that is worse than optimal by at most an $O(n)$ additive term, then we need to choose $\epsilon = \Omega(\log n)$. The\ next theorem shows that this is tight. The proof is deferred to the full version.

\begin{theorem} \label{thm:lbmultiitem}
    Suppose $M$ is an $\epsilon$-differentially private mechanism for
    the multi-item auction problem and the expected welfare achieve by $M$
    is at least $\opt - \frac{n}{10}$. Then $\epsilon = \Omega(\log n)$.
\end{theorem}

Note that in this  theorem, we do not restrict $M$ to be incentive
compatible. In other word, this lower bound holds for arbitrary
differentially private mechanisms. So there is no extra cost for 
imposing the truthfulness constraint.

\paragraph{Implication in BIC Blackbox Reduction} 

Recently, Hartline et al.~\cite{hartline2011bayesian} and Bei and Huang \cite{bei2011bayesian} introduce blackbox reductions that convert any algorithm into nearly Bayesian incentive-compatible mechanisms with only a marginal loss in the social welfare. Both approach essentially create a virtual interface for each agent which has the structure of a matching market and then run VCG in the virtual matching markets. By running the exponential mechanism instead of the VCG mechanism, we can obtain a blackbox reduction that converts any algorithm into a nearly Bayesian incentive-compatible and differentially private mechanism. We will defer more details to the full version of this paper.

\subsection{Procurement Auction for Spanning Trees}

Another interesting application is the procurement auction for a spanning tree (e.g., see \cite{cary2008auctions}). Procurement auctions (also known as reverse auctions) are a type of auction where the roles of buyers and sellers are reversed. In other word, the central entity seeks to buy, instead of sell, items or services from the agents. In particular in the procurement auction for spanning trees, consider $n = {k \choose 2}$ selfish agents own edges in a publicly known network of $k$ nodes. We shall imagine the nodes to be cities and the edges as potential highways connecting cities. Each agent $i$ has a non-negative cost $c_i$ for building a highway along the corresponding edge. The central entity (e.g., the government) wants to purchase a spanning tree from the network so that she can build highways to connect the cities. The goal is to design incentive compatible and differentially private mechanisms that provide good social welfare (minimizing total cost). 

Although this is a reverse auction in which agents have costs instead
of values and the payments are from the central entity to the
agents, by interpreting the costs as the negative of the valuations
(i.e.~$v_i = -c_i$ if the edge is purchased and $v_i = 0$ otherwise), we
can show that the exponential mechanism with the same payment scheme is
incentive compatible for procurement auctions via almost identical
proofs. We will omit the details in this extended abstract. 

Next, we will discuss how to efficiently implement the exponential
mechanism.

\paragraph{Sampling Spanning Trees} There has been a large body of
literature on sampling spanning tree (e.g., see
\cite{kulkarni1990generating} and the reference therein). Recently,
Asadpour et al.~\cite{asadpour2010log} have developed a polynomial time
algorithm for
sampling {\em entropy-maximizing} distributions, which is exactly the kind of
distribution used by the exponential mechanism. Therefore, the
allocation rule of the exponential mechanism can be implemented in
polynomial time for the spanning tree auction.

\paragraph{Implicit Payment Scheme by Babaioff, Kleinberg, and Slivkins \cite{babaioff2010truthful}} 
Although we can efficiently generate samples from the desired
distribution, it is not clear how to compute the exact payment
explicitly. Fortunately, Babaioff et
al.~\cite{babaioff2010truthful, bobby} provide a general method of 
computing an unbiased estimator for the payment given any rationalizable
allocation rule\footnote{Although the result in
\cite{babaioff2010truthful} only applies to single-parameter problems,
Kleinberg \cite{bobby} pointed out the same approach can be extended to
multi-parameter problems if the type space is convex.}. Hence, we
can use the implicit payment method in \cite{babaioff2010truthful,
bobby} to generate the payments in polynomial time.

\medskip 

Note that the size of the range of feasible outcomes of spanning tree
auction is the number of different spanning tree in a complete graph
with $k$ vertices, which equals $k^{k-2}$. By \pref{thm:expo} we have
the following:

\begin{theorem} \label{thm:tree}
    For any $\epsilon > 0$, the exponential mechanism
    $\expo^{\textrm{tree}}_\epsilon$ runs in polynomial time (in $k$ and
    $\epsilon^{-1}$), is IC, $\epsilon$-differentially private, and
    ensures that 
    \begin{align*}
      \Pr\left[\sum_{i=1}^n c_i \left( \widehat{\expo}^{\text{tree}}_\epsilon \right) > \opt + \frac{(k-2) \log k}{\epsilon} + \frac{t}{\epsilon} \right] \\
      \le \exp(-t) \enspace.
    \end{align*}
\end{theorem}

This trade-off between privacy and social welfare in \pref{thm:tree}
essentially means that we need $\epsilon =
\Omega(\log k)$ in order to get $\opt + O(k)$ guarantee on expected
total cost.  The next theorem 
shows that this tradeoff is also tight.  
The proof is deferred to the full version due to space
constraint.

\begin{theorem} \label{thm:lbspanningtree}
    Suppose $M$ is an $\epsilon$-differentially private mechanism for
    the procurement auction for spanning tree and the expected total
    cost by $M$ is at most $\opt + \frac{k}{24}$. Then 
    $\epsilon = \Omega(\log k)$.
\end{theorem}

Similar to the case in the multi-item auction, the above lower bound
does not restrict $M$ to be incentive compatible. So the exponential
mechanism is optimal even if we compare it to non-truthful ones.

\section*{Acknowledgement}

The authors would like to thank Aaron Roth for many useful comments and
helpful discussions.

\bibliographystyle{abbrv}
\bibliography{mechdesign}

\appendix

%{\bf Zhiyi: The above paragraph shall be placed at a proper place in the
%introduction. Maybe as a footnote of our discussion of payments and
%individual rationality.}

In this section, we will discuss what is the amount of noise one
needs to add to the payments in order to achieve $\epsilon$-differential
privacy. We will consider two different models depending on how the
payments are implemented: the {\em public payment} model and the {\em
private payment} model. 

In the public payment model, the payments of the agents will become
public information at the end of the auction, that is, the adversary who
tries to learn the private valuations of the agents can see all the
payments. Therefore, a payment scheme is $\epsilon$-differentially
private in the public payment model if and only if for any $i \in [n]$,
any value profiles $\bv = (v_1, \dots, v_n)$ and $\bv' = (v_1, \dots,
v'_i, \dots, v_n)$ that differ only in the valuation of agent $i$, and
any possible payment profile $\bp$, the probability 
\begin{align*}
  & \Pr[p_1(\bv), \dots, p_n(\bv) = \bp] \\
  \le ~ & \exp(\epsilon) \Pr[p_1(\bv'), \dots, p_n(\bv') = \bp] \enspace.
\end{align*}

In the private payment model, we will assume the payments are
implemented via secure channels such that the payment of each agent is
only known to the corresponding agent and a few trusted parties,
e.g., the central entity who runs the mechanism and/or the bank. Here, there are two cases based on what information
the adversary can learn from the payments. If the adversary is not one
of the agents, then by our assumption, he cannot see any of the payments
and therefore cannot learn any information from the payments. If the
adversary is one of the agents, then the only information of the
payments that he will have access to is his own payment. Therefore, a
payment scheme is $\epsilon$-differentially private in the public
payment model if and only if for any $i \ne j \in [n]$, any value
profiles $\bv = (v_1, \dots, v_n)$ and $\bv' = (v_1, \dots, v'_i, \dots,
v_n)$ that differ only in the valuation of agent $i$, and any possible
payment $p$ of agent $j$, the probability 
$$\Pr[p_j(\bv) = p] \le \exp(\epsilon) \Pr[p_j(\bv') = p] \enspace.$$

We will measure the amount of noise in the payments using $L_2$ norm,
that is, we aim to minimize the total variance of the agents'
payments in the worst-case: $\max_{\bv} \sum_{i=1}^n \Var[p_i(\bv)]$.

Next, we will proceed to analyze the amount of noise needed in each of
the two models. We will start with an upper bound on the sensitivity of
each agent's payment as a function of the bids.

\begin{lemma} \label{lem:paymentsensitivity}
    For any $i, j \in [n]$, and any value profiles $\bv = (v_1, \dots,
    v_n)$ and $\bv' = (v_1, \dots, v'_i, \dots, v_n)$ that only differ
    in the valuation of agent $i$, we have $|p_j(\bv) - p_j(\bv')| \le 1$.
\end{lemma}

\begin{proof}
    Note that by \pref{thm:main}, the exponential mechanism is individual
    rational. It is also easy to see that it has no positive transfer
    for that otherwise the zero-value agent could gain by lying. So by
    our assumption that the agents' valuations are always between $0$
    and $1$, we have $0 \le p_j(\bv), p_j(\bv') \le 1$. So
    \pref{lem:paymentsensitivity} follows trivially.
\end{proof}

In the public payment model, the mechanism has to reveal a vector of
$n$ real numbers (the payments) at the end of the auction, where each
entry has sensitivity $1$ by \pref{lem:paymentsensitivity}. Therefore,
we can use the standard treatment for answering numerical queries,
namely, adding independent Laplace noise $\Lap(\frac{n}{\epsilon})$ to
each entry, where $\Lap(b)$ is the Laplace distribution with
p.d.f.~$f_{\Lap(b)}(x) = \frac{1}{2b} \exp\left(-\frac{|x|}{b}\right)$.
More precisely, we can show the following theorem.

\begin{theorem}
    In the public payment model, the following payment scheme is
    $\epsilon$-differentially private and has total variance $O(n^{3/2}
    \epsilon^{-1})$, while maintaining the IC and IR
    %, and no-positive-transfer properties 
    in expectation: let $p_1, \dots,
    p_n$ be the payments specified in the exponential mechanism
    (\pref{fig:expo}); let $x_1, \dots, x_n$ be i.i.d.~variables
    following the Laplace distribution $\Lap(\frac{n}{\epsilon})$;
    use payment scheme $(p_1 + x_1, \dots, p_n + x_n)$.
\end{theorem}

The proof follows by standard analysis of the Laplace mechanism
(e.g., see \cite{dwork2006calibrating}). So we will omit the details in
this extended abstract. It is worth mentioning
that since the problem of designing payment scheme in the public payment
model is a special case of answering $n$ non-linear numerical queries,
it may be possible to reduce the amount of noise by using 
more specialized scheme on a problem-by-problem basis. However, we feel
this is less insightful than the other results we have in this paper, so
we will focus on general mechanisms and payment schemes that
work for all mechanism design problems.

Now let us turn to the private payment model. By our previous
discussion, the mechanism only need to release at most one real number 
to each potential adversary in this model. So one may expect much less
noise is needed in this model. Indeed, we could again use the standard
treatment of adding Laplace noise, but this time it suffices to add
independent Laplace noise $\Lap(\frac{1}{\epsilon})$ to each entry.

\begin{theorem}
    In the private payment model, the following payment scheme is
    $\epsilon$-differentially private and has total variance
    $O(\sqrt{n} \epsilon^{-1})$, while maintaining the IC and IR:
    %, and no-positive-transfer properties 
    in expectation: let $p_1, \dots,
    p_n$ be the payments specified in the exponential mechanism
    (\pref{fig:expo}); let $x_1, \dots, x_n$ be i.i.d.~variables
    following the Laplace distribution $\Lap(\frac{1}{\epsilon})$;
    use payment scheme $(p_1 + x_1, \dots, p_n + x_n)$.
\end{theorem}

\section{Approximate Implementation for Multi-Item Auction}
\label{app:implement}

In this section, we will explain how to approximately implement the
exponential mechanism in the multi-item auction setting. 
The main technical tool in this section is the seminal work of Jerrem,
Sinclair, and Vigoda \cite{jerrum2004polynomial} on approximating the
permanent of non-negative matrices, which can be phrased as follows:

\begin{lemma} [FPRAS for permanent of non-negative matrices 
    \cite{jerrum2004polynomial}] \label{lem:fpras}
    For any $\gamma > 0$ and any $\delta \in (0, 1)$, there is an
    algorithm that computes the permanent of an arbitrary $n \times n$
    matrix $A = \{a_{ij}\}_{i, j \in [n]}$ up to a multiplicative factor
    of $\exp(\gamma)$ with probability at least $1 - \delta$. The
    running time is polynomial in $n$, $\gamma^{-1}$, $\log
    (\delta^{-1})$, and $\log(\max_{i,j\in[n]} a_{ij} / \min_{i,j\in[n]}
    a_{ij})$.
\end{lemma}

To see the connection between the permanent of non-negative matrices and
implementation of the exponential mechanism in the multi-item auction
setting, we point out that the normalization factor in the outcome
distribution of the exponential mechanism is the permanent of a
non-negative matrix:
$$\sum_{r \in R_M} \exp\left( \frac{\epsilon}{2} \sum_{i=1}^n v_i(r)
\right) = \sum_{\pi \in \Pi_n} \prod_{i=1}^n \exp\left(
\frac{\epsilon}{2} v_{i\pi[i]} \right) = \textrm{perm}\left(\left\{
\exp\left(\frac{\epsilon}{2} v_{ij} \right) \right\}_{i, j \in
[n]}\right)\enspace.$$

We will let $A(\bv)$ denote the matrix 
$\{\exp(\frac{\epsilon}{2} v_{ij})\}_{i, j \in [n]}$. Moreover, we let
$A_{-i,-j}(\bv)$ denote the $(n-1)\times(n-1)$ matrix obtained by removing
the $i^{th}$ row and the $j^{th}$ column of $A(\bv)$.

\subsection{Approximate Sampler}

Now we are ready to introduce the approximate sampler for the multi-item
auction.

\begin{lemma} \label{lem:approxsampler}
    For any $\delta \in (0, 1)$ and $\gamma > 0$, there is a sampling
    algorithm whose running time is polynomial in $n$, $\epsilon^{-1}$
    $\gamma^{-1}$, and $\log\delta^{-1}$, such that with probability at
    least $1 - \delta$, it chooses an outcome $r$
    with probability
    $$\Pr[r] \in [\exp(-\gamma), \exp(\gamma)] \Pr[\expo^{R_M}_\epsilon
    = r] \enspace.$$
\end{lemma}

\begin{proof}
    We will recursively decide which item we will allocate to agent $i$
    for $i = 1, 2, \dots, n$ by repeatedly computing an accurate
    estimation of the marginal distribution. Concretely, the algorithm
    is given as follows:
    \begin{enumerate}
        \item Use the FPRAS in \pref{lem:fpras} to compute 
            $\textrm{perm}(A_{-1,-j}(\bv))$ up to a multiplicative factor
            of $\exp(\frac{\gamma}{2n})$ with success probability at
            least $1 - \frac{\delta}{n^2}$. Let $x_j$ denote the
            approximate value.
        \item Sample an item $j$ with probability $\Pr[j] \propto x_j$.
        \item Allocate item $j$ to agent $1$ and recurse on the
            remaining $n-1$ agents and $n-1$ items.
    \end{enumerate}

    First we note that for each allocation $\pi \in \Pi_n$,
    the probability that $\pi$ is chosen as the outcome can
    be decomposed into $n$ stages by Bayes' rule:
    \begin{align*}
        \Pr[\expo^{R_M}_\epsilon(\bv) = \pi] \, = \, &
        \Pr\big[\text{agent $1$ gets $\pi[1]$}\,\big] \cdot
        \Pr\big[\text{agent $2$ gets $\pi[2]$} \,|\, \pi[1]\,\big] \\
        & \cdots \Pr\big[\text{agent $n$ gets $\pi[n]$} \,|\, \pi[1],
        \dots, \pi[n-1]\,\big] \enspace.
    \end{align*}

    In the first recursion of our algorithm, we use the distribution
    $$\Pr[\text{agent $1$ gets item $j$}] \propto x_j \approx
    \textrm{perm}(A_{-1, -j}(\bv)) \enspace.$$
    
    Further, in the exponential mechanism
    \begin{eqnarray*}
        \Pr[\text{agent $1$ gets item $j$ in $\expo^{R_M}_\epsilon$}]
        & \propto & \sum_{\pi : \pi[1] = j} \exp\left(
        \frac{\epsilon}{2} \sum_{k=1}^n v_{k\pi[k]} \right) \\
        & = & \exp\left(\frac{\epsilon}{2} v_{1j}\right)
        \textrm{perm}(A_{-1, -j}(\bv)) \enspace.
    \end{eqnarray*}

    Since $x_j$ approximate $\textrm{perm}(A_{-1,-j}(\bv))$ up to an
    $\exp(\frac{\gamma}{2n})$ factor, we know the probability that item
    $j$ is allocated to agent $1$ in our algorithm approximate the
    correct marginal up to an $\exp(\frac{\gamma}{n})$ multiplicative
    factor.

    Similar claim holds for the rest of the $n-1$ stages as well. So the
    probability that we samples a permutation $\pi \in R_M$ differs from
    the correct distribution by at most a $\exp(\frac{\gamma}{n})^n =
    \exp(\gamma)$ factor. Moreover, by union bound the
    failure probability is at most $\delta$.
\end{proof}

\subsection{Approximate Payments}

Next, we will turn to approximate implementation of the payment
scheme. First, recall that the payment for agent $i$ is
\begin{eqnarray*}
    p_i & = & \E_{r \sim \expo^{R_M}_\epsilon(\bv)} [v_i(r)] -
    \frac{2}{\epsilon} \ln \left(\sum_{r \in R_M} \exp\left(
    \frac{\epsilon}{2} \sum_{k=1}^n v_k(r) \right) \right) +
    \frac{2}{\epsilon} \ln \left( \sum_{r \in R_M} \exp \left(
    \frac{\epsilon}{2} \sum_{k \ne i} v_k(r) \right) \right) \\
    & = & \E_{r \sim \expo^{R_M}_\epsilon(\bv)} [v_i(r)] -
    \frac{2}{\epsilon} \ln \left(\textrm{perm}\left(A(v_i,
    v_{-i})\right) \right) + \frac{2}{\epsilon} \ln
    \left(\textrm{perm}\left(A(\bm{0}, v_{-i})\right) \right) 
    \enspace.
\end{eqnarray*}

The next lemma states that we can efficiently compute an estimator for
the payment $p_i$ with inverse polynomially small bias.

\begin{lemma}
    For any $\delta \in (0, 1)$ and $\gamma \in (0, 1)$, we can compute
    in polynomial time (in $n$, $\epsilon^{-1}$, and $\gamma^{-1}$) a
    random estimator $\hat{p}_i$ for $p_i$ such that the bias is small:
    $\abs{\E[\hat{p}_i] - p_i} \le \gamma$.
\end{lemma}

\begin{proof}
    By \pref{lem:fpras}, we can efficiently estimate
    $\textrm{perm}(A(v_i, v_{-i}))$ and $\textrm{perm}(A(\bm{0},
    v_{-i}))$ up to an multiplicative factor of $\exp(\frac{\gamma}{6})$
    with success probability at least $1 - \frac{\gamma}{6}$. Hence, we
    can compute $\ln(\textrm{perm}(A(v_i, v_{-i}))$ and
    $\ln(\textrm{perm}(A(\bm{0}, v_{-i})))$ up to additive bias of
    $\frac{\gamma}{6}$ with probability $1 - \frac{\gamma}{6}$. Note
    that the total bias introduced if the FPRAS fails is at most $1$ and
    that could happens with probability at most $\frac{\gamma}{6}$. So
    the total bias from estimating $\ln(\textrm{perm}(A(v_i, v_{-i})))$
    and $\ln(\textrm{perm}(A(\bm{0}, v_{-i})))$ is at most
    $\frac{\gamma}{2}$.

    It remains to compute an estimator for $\E_{r \sim
    \expo^{R_M}_\epsilon(\bv)}[v_i(r)]$ with bias less than
    $\frac{\gamma}{2}$. In order to do so, we will use the algorithm in
    \pref{lem:approxsampler} to sample an outcome $r^*$ from a
    distribution whose probability mass function differs from that of
    $\expo^{R_M}_\epsilon(\bv)$ by at most a $\exp(\frac{\gamma}{6})$
    factor point-wise, with success probability at least $1 -
    \frac{\gamma}{6}$. Then we will use $v_i(r^*)$ as our estimator.
    Note that conditioned on the sampler runs correctly, we have
    $$\abs{\E[v_i(r^*)] - \E_{r \sim \expo^{R_M}_\epsilon(\bv)} [v_i(r)]}
    \le \left(\exp\left(\frac{\gamma}{6}\right) - 1\right) \E_{r \sim
    \expo^{R_M}_\epsilon(\bv)}
    [v_i(r)] \le \left(\exp(\frac{\gamma}{6}) - 1\right) \le
    \frac{\gamma}{3} \enspace.$$

    Moreover, the maximum bias conditioned on the failure of the sampler
    is at most $1$, which happens with probability at most
    $\frac{\gamma}{6}$. So the total bias from the estimator for 
    $\E_{r \sim \expo^{R_M}_\epsilon(\bv)}[v_i(r)]$ is at most
    $\frac{\gamma}{2}$.
\end{proof}

\section{Lower Bound for Multi-Item Auction} \label{sec:lbmultiitem}

\begin{proof}[Proof of \pref{thm:lbmultiitem}] 
    Let us first define some notations. For any $j^* \in [n]$, we will
    let $e^{j^*}$ denote the valuation profile such that $e^{j^*}_j = 1$
    if $j = j^*$ and $e^{j^*}_j = 0$ if $j \ne j^*$. That is, an agent
    with valuation $e^{j^*}$ is single-minded who only value getting
    item $j^*$ (with value $1$) and has no interest in getting any other
    item. We will say $j^*$ is the {\em critical item} for this agent.

    Suppose $M$ is an $\epsilon$-differentially private mechanism such
    that $M$ always obtain at least $\opt - \frac{n}{10}$ expected
    social welfare. Let us consider the following randomly chosen
    instance: each agent's valuation is chosen from $e^1, \dots, e^n$
    independently and uniformly at random. Let us consider the social
    welfare we get by
    running mechanism $M$ on this randomly constructed instance. We
    first note that $\E_{\bv}[\opt(\bv)] = (1 - e^{-1}) n$ for that each
    item has probability $1 - e^{-1}$ of being the critical item of
    at least one of the agents. By our assumption, the expected welfare
    obtained by $M$ shall be at least $(1 - e^{-1})n - \frac{n}{10} >
    \frac{n}{2}$. Therefore, we have 
    $$\sum_{i=1}^n \sum_{j=1}^n \Pr[\text{$M$ allocate $j$ to agent
    $i$} \,|\, \text{$j$ is critical for $i$}] \Pr[\text{$j$ is
    critical for $i$}] \ge \frac{n}{2} \enspace.$$

    Note that $\Pr[\text{$j$ is critical for $i$}] = \frac{1}{n}$ for
    all $i, j \in [n]$, we get that the average probability that a
    critical item-agent pair is allocated is at least half: 
    \begin{equation} \label{eq:lbmatching1} 
        \frac{1}{n^2} \sum_{i=1}^n \sum_{j=1}^n \Pr[\text{$M$ allocate
        $j$ to agent $i$} \,|\, \text{$j$ is critical for $i$}] \ge
        \frac{1}{2} \enspace.  
    \end{equation}

    Similarly, we have
    $$\sum_{i=1}^n \sum_{j=1}^n \Pr[\text{$M$ allocate $j$
    to agent $i$} \,|\, \text{$j$ is not critical for $i$}]
    \Pr[\text{$j$ is not critical for $i$}] \le \frac{n}{2} \enspace.$$

    Note that $\Pr[\text{$j$ is not critical for $i$}] = \frac{n-1}{n}$
    for all $i, j \in [n]$, we get that the average probability that
    the average probability that a non-critical
    item-agent pair is chosen in the allocation is very small: 
    \begin{equation} \label{eq:lbmatching2}
        \frac{1}{n^2} \sum_{i=1}^n \sum_{j=1}^n \Pr[\text{$M$ allocate
        $j$ to agent $i$} \,|\, \text{$j$ is not critical for $i$}] \le
        \frac{1}{2(n-1)} \enspace.
    \end{equation}

    By \pref{eq:lbmatching1} and \pref{eq:lbmatching2}, we have
    $$\frac{ \sum_{i=1}^n \sum_{j=1}^n \Pr[\text{$M$ allocate $j$ to
    agent $i$} \,|\, \text{$j$ is critical for $i$}]}
    {\sum_{i=1}^n \sum_{j=1}^n \Pr[\text{$M$ allocate $j$ to agent $i$}
    \,|\, \text{$j$ is not critical for $i$}]} \ge n-1
    \enspace.$$

    In particular, we know there exists a $(i, j)$
    pair such that
    $$\frac{ \Pr[\text{$M$ allocate $j$ to
    agent $i$} \,|\, \text{$j$ is critical for $i$}]}
    { \Pr[\text{$M$ allocate $j$ to agent $i$}
    \,|\, \text{$j$ is not critical for $i$}]} \ge n-1
    \enspace.$$
    
    Since $M$ is
    $\epsilon$-differentially private, we get that $\exp(\epsilon) \ge
    n-1$, and thus $\epsilon = \Omega(\log n)$.
\end{proof}

\section{Lower Bound for Procurement Auction for Spanning Trees}
\label{app:lbspanningtree}

\begin{proof}[Proof of \pref{thm:lbspanningtree}]
    Suppose $M$ is an $\epsilon$-differentially private mechanism whose
    expected total cost is at most $\opt + \frac{k}{24}$.

    We will consider the following randomly generated instance. Each
    agent $i$'s cost value $c_i$ is independently chosen as 
    $$c_i = \left\{ \begin{aligned}
        & 1 & & \text{ , w.p. } 1 - \frac{1}{2k} \\
        & 0 & & \text{ , w.p. } \frac{1}{2k}
    \end{aligned}\right. $$

    If an agent has cost $0$, we say this agent and the corresponding
    edge are {\em critical}. Let us first analyze the expected value of
    $\opt$ for such randomly generated instances. Intuitively, we want
    to pick as many critical edges as possible. In particular, when
    there are no cycles consists of only critical edges, the minimum
    spanning tree shall pick all critical edges, which comprise a forest
    in the graph, and then pick some more edges to complete the spanning
    tree.

    \begin{lemma} \label{lem:st1}
        With probability at least $\frac{1}{2}$, there are no cycle
        consists of only critical edges.
    \end{lemma}
    
    \begin{proof}[Proof of \pref{lem:st1}]
        For each cycle of length $t$, the probability that all edges on
        this cycle are critical is $(2k)^{-t}$. 
        Note that the number of cycles of
        length $t$ is at most ${k \choose t} (t-1)! \le k^t$. Here ${k
        \choose t}$ is the number of subsets of $t$ vertices
        and $(t-1)!$ is the number of different Hamiltonian cycles among
        $t$ vertices. Hence, by union bound, the probability that there
        is any cycle consists of only critical edges is at most
        $\sum_{t=2}^k (2k)^{-t} \cdot k^t = \sum_{t=2}^k 2^{-t} <
        \frac{1}{2}$.
    \end{proof}

    Moreover, by Chernoff-H\"oeffding bound, we have that the number of
    critical edges is at least $\frac{k}{3}$ with probability at least
    $\frac{3}{4}$. 
    
    Therefore, by union bound, with probability at least
    $\frac{1}{4}$, we have that there are at least $\frac{k}{3}$
    critical edges and there are no cycle consists of only critical
    edges. So in this case, we have $\opt \le k - \frac{k}{3} =
    \frac{2k}{3}$. Therefore, the expectation of the optimal total cost
    is at most $\E[\opt] \le \frac{3}{4} k + \frac{1}{4} \frac{2k}{3} =
    \frac{11k}{12}$.
    
    By our assumption on $M$, we get that the expected total cost of the
    outcome chosen by $M$ is at most $\frac{11k}{12} + \frac{k}{24} =
    \frac{23k}{24}$. In other words, the expected number of critical
    edges chosen by $M$ is at least $\frac{k}{24}$. That is,
    $$\sum_{i=1}^n \Pr[\text{edge $i$ is chosen} \,|\, \text{edge $i$ is
    critical}] \Pr[\text{edge $i$ is critical}] \ge \frac{k}{24}
    \enspace.$$

    Note that $\Pr[\text{edge $i$ is critical}] = \frac{1}{2k}$ for all
    $i \in [n]$ and $n = {k \choose 2} = \frac{k(k-1)}{2}$, we get that
    on average a critical edge is chosen with at least constant
    probability
    $$\frac{1}{n} \sum_{i=1}^n \Pr[\text{edge $i$ is chosen} \,|\,
    \text{edge $i$ is critical}] \ge \frac{1}{6} \enspace.$$

    On the other hand, it is easy to see
    $$\sum_{i=1}^n \Pr[\text{edge $i$ is chosen} \,|\, \text{edge $i$ is
    not critical}] \Pr[\text{edge $i$ is not critical}] \le k
    \enspace.$$

    By $\Pr[\text{edge $i$ is not critical}] = 1 - \frac{1}{2k}$ and $n
    = {k \choose 2}$, we get that on average a non-critical edge is chosen with very small probability
    $$\frac{1}{n} \sum_{i=1}^n \Pr[\text{edge $i$ is chosen} \,|\,
    \text{edge $i$ is not critical}] \le \frac{2k^2}{(2k-1)n} =
    \frac{4k}{(k-1)(2k-1)} \le \frac{8}{2k-1} \enspace.$$

    Therefore, we have
    $$\frac{\sum_{i=1}^n \Pr[\text{edge $i$ is chosen} \,|\,
    \text{edge $i$ critical}]}{\sum_{i=1}^n \Pr[\text{edge $i$ is
    chosen} \,|\, \text{edge $i$ is not critical}]} \ge \frac{2k-1}{48}
    \enspace.$$

    In particular, there exists an agent $i$, such that
    $$\frac{\Pr[\text{edge $i$ is chosen} \,|\, \text{edge $i$
    critical}]}{\Pr[\text{edge $i$ is chosen} \,|\, \text{edge $i$ is
    not critical}]} \ge \frac{2k-1}{48} \enspace.$$

    However, the above amount is upper bounded by $\exp(\epsilon)$ since
    $M$ is $\epsilon$-differentially private. So we conclude that
    $\epsilon = \Omega(k)$.
\end{proof}

\end{document}